\documentclass[graybox]{svmult}

\usepackage{epsfig}

\usepackage{mathptmx}       
\usepackage{helvet}         
\usepackage{courier}        
\usepackage{type1cm}        
\usepackage{makeidx}         
\usepackage{graphicx}        
\usepackage{multicol}        
\usepackage[bottom]{footmisc}

\newcommand{\be}{\begin{equation}}
\newcommand{\ee}{\end{equation}}
\newcommand{\bea}{\begin{eqnarray}}
\newcommand{\eea}{\end{eqnarray}}
\newcommand{\ba}{\begin{array}}
\newcommand{\ea}{\end{array}}

\makeindex             

\begin{document}

\title*{MIC: Mutual Information based hierarchical Clustering}
\author{Alexander Kraskov and Peter Grassberger}
\institute{Alexander Kraskov \at UCL Institute of Neurology, Queen Square, London, WC1N 3BG, UK, \email{akraskov@ion.ucl.ac.uk}
\and Peter Grassberger \at Department of Physics \& Astronomy and Institute for Biocomplexity \& Informatics, 
University of Calgary, 2500 University Drive NW, Calgary AB, Canada T2N 1N4, \email{pgrassbe@ucalgary.ca}}
%
%
\maketitle

\abstract{Clustering is a concept used in a huge variety of
applications. We review a conceptually very simple algorithm
for hierarchical clustering called in the following the {\it mutual information
clustering} (MIC) algorithm. It uses mutual information (MI) as a similarity measure and exploits its grouping
property: The MI between three objects $X, Y,$ and $Z$ is equal to the sum of the MI between $X$ and $Y$, plus
the MI between $Z$ and the combined object $(XY)$.
We use MIC both in the Shannon (probabilistic) version of information theory, where
the ``objects" are probability distributions represented by random samples, and in the Kolmogorov (algorithmic)
version, where the ``objects" are symbol sequences. We apply our method to the construction of
phylogenetic trees from mitochondrial DNA sequences and we reconstruct the fetal ECG from the output of
independent components analysis (ICA) applied to the ECG of a pregnant woman. 
}

\section{Introduction}
\label{sec:intro}
Classification or organizing of data is a crucial task in all scientific disciplines. It is one of the most
fundamental mechanism of understanding and learning \cite{jain-dubes}. Depending on the problem, classification
can be exclusive or overlapping, supervised or unsupervised. In the following we will be interested only in
exclusive unsupervised classification. This type of classification is usually called clustering or cluster
analysis.

An instance of a clustering problem consists of a set of objects and a set of properties (called characteristic
vector) for each object. The goal of clustering is the separation of objects into groups using only the
characteristic vectors. Indeed, in general only certain aspects of the characteristic vectors will be relevant,
and extracting these relevant features is one field where mutual information (MI) plays a major role
\cite{bottleneck}, but we shall not deal with this here. Cluster analysis organizes data either as a single
grouping of individuals into non-overlapping clusters or as a hierarchy of nested partitions. The former approach
is called partitional clustering (PC), the latter one is hierarchical clustering (HC). One of the main features
of HC methods is the visual impact of the tree or {\it dendrogram} which enables one to see how objects are being 
merged into clusters. From any HC one can obtain a PC by restricting oneself to a ``horizontal" cut through the
dendrogram, while one cannot go in the other direction and obtain a full hierarchy from a single PC. Because of
their wide spread of applications, there are a large variety of different clustering methods in use
\cite{jain-dubes}. In the following we shall only deal with {\it agglomerative} hierarchical clustering, where 
clusters are built by joining first the most obvious objects into pairs, and then continues to join build up 
larger and larger objects. Thus the tree is built by starting at the leaves, and is finished when the main 
branches are finally joined at the root. This is opposed to algorithms where one starts at the root and splits 
clusters up recursively. In either case, the tree obtained in this way can be refined later by restructuring 
it, e.g. using so-called {\it quartet methods} \cite{strimmer,cilibrasi06}.

The crucial point of all clustering algorithms is the choice of a {\it proximity measure}. This is obtained from
the characteristic vectors and can be either an indicator for similarity (i.e. large for similar and small for
dissimilar objects), or dissimilarity. In the latter case it is convenient (but not obligatory) if it satisfies
the standard axioms of a {\it metric} (positivity, symmetry, and triangle inequality). A matrix of all pairwise
proximities is called proximity matrix. Among agglomerative HC methods one should distinguish between those where one uses the
characteristic vectors only at the first level of the hierarchy and derives the proximities between clusters
from the proximities of their constituents, and methods where the proximities are calculated each time from
their characteristic vectors. The latter strategy (which is used also in the present paper) allows of course for
more flexibility but might also be computationally more costly. There exist a large number of different
strategies \cite{jain-dubes,press}, and the choice of the optimal strategy depends on the characteristics of 
the similarities: for ultrametric distances, e.g., the ``natural" method is UPGMA \cite{press}, while 
{\it neighbor joining} is the natural choice when the distance matrix is a metric satisfying the four-point 
condition $d_{ij}+d_{kl} \leq \max(d_{ik}+d_{jl},d_{il}+d_{jk})$ \cite{press}. 

In the present chapter we shall use proximities resp. distances derived from mutual information \cite{cover-thomas}. 
In that case the distances neither form an ultrametric, nor do they satisfy the above four-point condition. Thus neither 
of the two most popular agglomerative clustering methods are favored. But we shall see that the distances 
have another special feature which suggests a different clustering strategy discussed first in \cite{kraskov03,kraskov05}.

Quite generally, the ``objects" to be clustered can be either single (finite) patterns (e.g. DNA sequences) or
random variables, i.e. {\it probability distributions}. In the latter case the data are usually supplied in form
of a statistical sample, and one of the simplest and most widely used similarity measures is the linear
(Pearson) correlation coefficient. But this is not sensitive to nonlinear dependencies which do not manifest
themselves in the covariance and can thus miss important features. This is in contrast to mutual information
(MI) which is also singled out by its information theoretic background \cite{cover-thomas}. Indeed, MI is zero
only if the two random variables are strictly independent.

Another important feature of MI is that it has also an ``algorithmic" cousin, defined within algorithmic
(Kolmogorov) information theory \cite{li-vi} which measures the similarity between individual objects. For a 
comparison between probabilistic and algorithmic information theories, see \cite{grunwald04}. For a
thorough discussion of distance measures based on algorithmic MI and for their application to clustering, see
\cite{li01,li04,cilibrasi05}.

Another feature of MI which is essential for the present application is its {\it grouping property}: The MI
between three objects (distributions) $X, Y,$ and $Z$ is equal to the sum of the MI between $X$ and $Y$, plus
the MI between $Z$ and the combined object (joint distribution) $(XY)$,
\be
   I(X,Y,Z) = I(X,Y) + I((X,Y),Z).                        \label{group}
\ee
Within Shannon information theory this is an exact theorem (see below), while it is true in the algorithmic
version up to the usual logarithmic correction terms \cite{li-vi}. Since $X,Y,$ and $Z$ can be themselves
composite, Eq.(\ref{group}) can be used recursively for a cluster decomposition of MI. This motivates the main
idea of our clustering method: instead of using e.g. centers of masses in order to treat clusters like
individual objects in an approximative way, we treat them exactly like  individual objects when using MI as
proximity measure.

More precisely, we propose the following scheme for clustering $n$ objects with MIC:\\
(1) Compute a proximity matrix based on pairwise mutual informations; assign $n$ clusters such that each cluster
contains exactly one object;\\
(2) find the two closest clusters $i$ and $j$; \\
(3) create a new cluster $(ij)$ by combining $i$ and $j$; \\
(4) delete the lines/columns with indices $i$ and $j$ from the proximity matrix, and add one line/column
containing the proximities between cluster $(ij)$ and all
other clusters. These new proximities are calculated by either treating $(X_i,X_j)$ as a single random variable
   (Shannon version), or by concatenating $X_i$ and $X_j$ (algorithmic version); \\
(5) if the number of clusters is still $>2$, goto (2); else join the two clusters and stop.


In the next section we shall review the pertinent properties of MI, both in the Shannon and in the algorithmic
version. This is applied in Sec.~\ref{sec:phylo} to construct a phylogenetic tree using mitochondrial DNA and 
in Sec.~\ref{sec:ecg} to
cluster the output channels of an independent component analysis (ICA) of an electrocardiogram (ECG) of a
pregnant woman, and to reconstruct from this the maternal and fetal ECGs. We finish with our conclusions in
Sec.~\ref{sec:concl}.

\section{Mutual Information}
\label{sec:mi}
\subsection{Shannon Theory}
\label{subsec:shan}
Assume that one has two random variables $X$ and $Y$. If they are discrete, we write $p_i(X) = {\rm
prob}(X=x_i)$, $p_i(Y) = {\rm prob}(Y=x_i)$, and $p_{ij} = {\rm prob}(X=x_i,Y=y_i)$ for the marginal and joint
distribution. Otherwise (and if they have finite densities) we denote the densities by $\mu_X(x),\mu_Y(y)$, and
$\mu(x,y)$. Entropies are defined for the discrete case as usual by $H(X) = - \sum_ip_i(X) \log p_i(X)$, $H(Y) =
- \sum_ip_i(Y) \log p_i(Y)$, and $H(X,Y)=-\sum_{i,j} p_{ij} \log p_{ij}$. Conditional entropies are defined as
$H(X|Y) = H(X,Y)-H(Y) = -\sum_{i,j} p_{ij} \log p_{i|j}$. The base of the logarithm determines the units in
which information is measured. In particular, taking base two leads to information measured in bits. In the
following, we always will use natural logarithms. The MI between $X$ and $Y$ is finally defined as
\bea
   I(X,Y) &=& H(X)+H(Y)-H(X,Y) \nonumber  \\
          &=& \sum_{i,j} p_{ij}\;\log{p_{ij}\over p_i(X)p_j(Y)}.
\eea
It can be shown to be non-negative, and is zero only when $X$ and $Y$ are strictly independent. For $n$ random
variables $X_1,X_2\ldots X_n$, the MI is defined as
\be
   I(X_1,\ldots, X_n) = \sum_{k=1}^n H(X_k) - H(X_1,\ldots, X_n).
\ee 
This quantity is often referred to as (generalized) redundancy, in order to distinguish it from different
``mutual informations" which are constructed analogously to higher order cumulants \cite{cover-thomas}, 
but we shall not follow this
usage. Eq.(\ref{group}) can be checked easily,
\bea
   I(X,Y,Z) &=& H(X)+H(Y)+H(Z) - H(X,Y,Z)  \nonumber  \\
          &=& \sum_{i,j,k} p_{ijk}\;\log{p_{ijk}\over p_i(X)p_j(Y)p_k(Z)} \nonumber  \\
          &=& \sum_{i,j,k} p_{ijk}\;[\log{p_{ij}(XY)\over p_i(X)p_j(Y)} + \log{p_{ijk}\over p_{ij}(XY)p_k(Z)}] \nonumber  \\
          &=& I(X,Y) + I((X,Y),Z),
\eea
together with its generalization to arbitrary groupings. It means
that MI can be {\it decomposed into hierarchical levels}. By iterating it, one can decompose $I(X_1\ldots X_n)$
for any $n>2$ and for any partitioning of the set $(X_1\ldots X_n)$ into the MIs between elements within one
cluster and MIs between clusters.

For continuous variables one first introduces some binning (`coarse-graining'), and applies the above to the
binned variables. If $x$ is a vector with dimension $m$ and each bin has Lebesgue measure $\Delta$, then $p_i(X)
\approx \mu_X(x)\Delta^m$ with $x$ chosen suitably in bin $i$, and \footnote{Notice that we have here assumed
that densities really exists. If not e.g. if $X$ lives on a fractal set), then $m$ is to be replaced by the
Hausdorff dimension of the measure $\mu$.}
\be
   H_{\rm bin}(X) \approx \tilde{H}(X) - m \log \Delta
\ee
where the {\it differential entropy} is given by
\be
   \tilde{H}(X) = -\int dx \;\mu_X(x) \log \mu_X(x).
\ee
Notice that $H_{\rm bin}(X)$ is a true (average) information and is thus non-negative, but $\tilde{H}(X)$ is not
an information and can be negative. Also, $\tilde{H}(X)$ is not invariant under homeomorphisms $x\to \phi(x)$.

Joint entropies, conditional entropies, and MI are defined as above, with sums replaced by integrals. Like
$\tilde{H}(X)$, joint and conditional entropies are neither positive (semi-)definite nor invariant. But MI,
defined as
\be
   I(X,Y) = \int\!\!\!\int dx dy \;\mu_{XY}(x,y) \;\log{\mu_{XY}(x,y)\over \mu_X(x)\mu_Y(y)}\;,
   \label{mi}
\ee
is non-negative and invariant under $x\to \phi(x)$ and $y\to \psi(y)$. It is (the limit of) a true information,
\be
   I(X,Y) = \lim_{\Delta\to 0} [H_{\rm bin}(X)+H_{\rm bin}(Y)-H_{\rm bin}(X,Y)].
\ee

\subsection{Estimating mutual Shannon information}

In applications, one usually has the data available in form of a statistical sample. To estimate $I(X,Y)$ one
starts from $N$ bivariate measurements $(x_i,y_i), \, i=1,\ldots N$ which are assumed to be iid (independent
identically distributed) realizations. For discrete variables, estimating the probabilities $p_i$, $p_{ij}$, etc.,
is straightforward: $p_i$ is just approximated by the ratio $n_i/N$, where $n_i$ is the number of outcomes 
$X=x_i$. This approximation gives rise both to a bias in the estimate of entropies, and to statistical 
fluctuations. The bias can be largely avoided by more sophisticated methods (see e.g. \cite{grass03}), but 
we shall not go into details.

For continuous variables, the situation is worse. There exist numerous strategies to estimate $I(X,Y)$. The 
most popular include discretization by partitioning the ranges of $X$ and $Y$ into finite intervals 
\cite{darbellay-vajda}, ``soft" or ``fuzzy" partitioning using B-splines \cite{daub}, and kernel density 
estimators \cite{moon95}. We shall use in the following the MI estimators based on $k$-nearest neighbors 
statistics proposed in Ref.~\cite{kraskov04}, and we refer to this paper for a comparison with alternative 
methods.

\subsection{$k$-nearest neighbors estimators}

There exists an extensive literature
on nearest neighbors based estimators for the simple Shannon entropy 
\be
   H(X) = -\int dx \mu(x) \log \mu(x),
\ee
dating back at least to \cite{dobrushin,vasicek}. But it 
seems that these methods have never been used for estimating MI. In 
\cite{vasicek,dude-meul,es,ebrahimi,correa,tsyb-meul,wiecz-grze} it is 
assumed that $x$ is one-dimensional, so that the $x_i$ can be ordered by
magnitude and $x_{i+1}-x_i \to 0$ for $N\to \infty$. In the simplest case, 
the estimator based only on these distances is 
\be 
   H(X) \approx -{1\over N-1} \sum_{i=1}^{N-1} \log(x_{i+1}-x_i) - \psi(1) +
        \psi(N)\;.
   \label{vasi}
\ee
Here, $\psi(x)$ is the digamma function, $\psi(x) = \Gamma(x)^{-1} d\Gamma(x)/dx$.
It satisfies the recursion $\psi(x+1) = \psi(x)+1/x$ and $\psi(1) = -C$ where 
$C = 0.5772156\ldots$ is the Euler-Mascheroni constant. For large $x$, 
$\psi(x) \approx \log x -1/2x$. Similar formulas exist which use $x_{i+k}-x_i$
instead of $x_{i+1}-x_i$, for any integer $k<N$.

Although Eq.(\ref{vasi}) and its generalizations to $k>1$ seem to give the 
best estimators of $H(X)$, they cannot be used for MI because it is not 
obvious how to generalize them to higher dimensions. Here we have to use a 
slightly different approach, due to \cite{koza-leon}.

Assume some metrics to be given on the spaces spanned by $X, Y$ and $Z=(X,Y)$. 
In the following we shall use always the maximum norm in the joint space, 
i.e. 
\be
   || z-z'|| = \max\{||x-x'||,||y-y'||\},
\ee
independently of the norms used for $||x-x'||$ and $||y-y'||$ (they need not 
be the same, as these spaces could be completely different).
We can then rank, for each point $z_i = (x_i,y_i)$, its neighbors by distance
$d_{i,j}= || z_i-z_j||$: $d_{i,j_1} \leq d_{i,j_2} \leq d_{i,j_3} \leq \ldots$.
Similar rankings can be done in the subspaces $X$ and $Y$. The basic idea of
\cite{koza-leon} is to estimate $H(X)$ from the average
distance to the $k$-nearest neighbor, averaged over all $x_i$. Mutual information could be obtained by estimating in this
way $H(X)$, $H(Y)$ and $H(X,Y)$ separately and using
\be
   I(X,Y) = H(X)+H(Y)-H(X,Y)\;.
\ee
But using the same $k$ in both the marginal and joint spaces would mean that 
the typical distances to the $k-$th neighbors are much larger in the joint ($Z$) space 
than in the marginal spaces. This would mean that any errors made in the individual estimates would 
presumably not cancel, since they depend primarily on these distances. 

Therefore we proceed differently in \cite{kraskov04}. We first choose a value of $k$, which gives the 
number of neighbors in the joint space. From this we obtain for each point 
$z_i = (x_i,y_i)$ a length scale $\epsilon_i$, and then we count the number of neighbors within 
this distance for each of the marginal points $x_i$ and $y_i$.

Indeed, for each $k$ two different versions of this algorithm were given in \cite{kraskov04}. 
In the first, neighborhoods in the joint space are chosen as (hyper-)squares, so that the the 
length scales $\epsilon_i$ are the same in $x$ and in $y$. In the second version, the size 
of the neighborhood is further optimized by taking them to be (hyper-)rectangles, so that 
$\epsilon_{i,x} \neq \epsilon_{i,y}$.
Also, the analogous estimators for the generalized redundancies $I(X_1,X_2, \ldots X_m)$ were 
given there. Both variants give very similar results. For details see Ref.~\cite{kraskov04}. 

Compared to other estimators, these estimators are of similar speed (they are faster than 
methods based on kernel density estimators, slower than the B-spline estimator of \cite{daub},
and of comparable speed to the sophisticated adaptive partitioning method of \cite{darbellay-vajda}.
They are rather robust (i.e., they are insensitive to outliers). Their superiority becomes 
most evident in higher dimensions, where any method based on partitioning fails. 
Any estimator has statistical errors (due to sample-to-sample fluctuations) and some bias.
Statistical errors decrease with $k$, while the bias increases in general with $k$. Thus it is advised 
to take large $k$ (up to $k/N \approx 0.1$) when the bias is not expected be a problem, 
and to use small $k$ ($k=1$, in the extreme), if a small bias is more important than small 
statistical sample-to-sample fluctuations.  

\begin{figure}
  \begin{center}
    \psfig{file=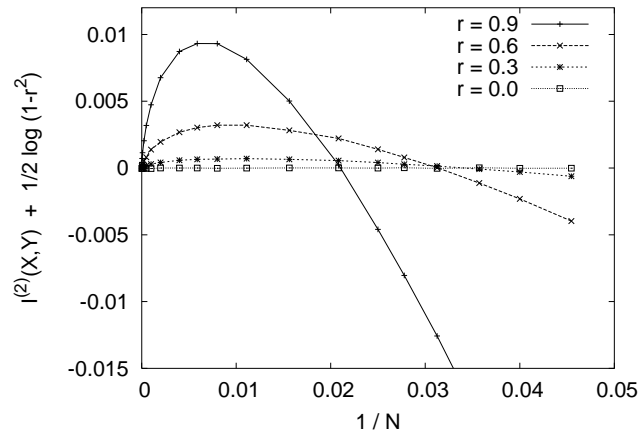,width=6.0cm,angle=270}
 \caption{Averages of the estimates of ${\hat I}^{(2)}(X,Y) - I_{\rm exact}(X,Y)$ for Gaussians with
   unit variance and covariances $r=0.9, 0.6, 0.3$, and 0.0 (from top to bottom),
   plotted against $1/N$. In all cases $k=1$. The number of realizations over which this is averaged is 
   $>2\times 10^6$ for $N<=1000$, and decreases to $\approx 10^5$ for $N=40,000$.
   Error bars are smaller than the sizes of the symbols.}
 \label{fig:I2gauss}
 \end{center}
 \end{figure}

A systematic study of the performance of these estimators and 
comparison with previous algorithms is given in Ref.~\cite{kraskov04}. Here we will discuss
just one further feature of the estimators proposed in \cite{kraskov04}: They seem to be 
strictly unbiased whenever the true mutual information is zero. This makes them particularly 
useful for a test for independence. We have no analytic proof for this, but very good 
numerical evidence. As an example, we show in Fig.~\ref{fig:I2gauss} results for Gaussian 
distributions. More precisely, we drew a large number (typically $10^6$ and more) of $N-$tuples 
of vectors $(x,y)$ from a bivariate Gaussian with fixed covariance $r$, and estimated 
$I(X,Y)$ with $k=1$ by means of the second variant ${\hat I}^{(2)}(X,Y)$ of our estimator. The 
averages over all tuples of ${\hat I}^{(2)}(X,Y) - I_{\rm Gauss}(X,Y)$ is plotted in Fig.~\ref{fig:I2gauss}
against $1/N$. Here, 
\be
   I_{\rm Gauss}(X,Y) = -{1\over 2} \log(1-r^2)\;.
   \label{gauss-MI}
\ee
is the exact MI for Gaussians with covariance $r$ \cite{darbellay-vajda}.

The most conspicuous feature seen in Fig.~\ref{fig:I2gauss}, apart from the fact that indeed
$I^{(2)}(X,Y) - I_{\rm Gauss}(X,Y) \to 0$ for $N\to \infty$, is that the 
systematic error is compatible with zero for $r=0$, i.e. when the two Gaussians
are uncorrelated. We checked this with high statistics runs for many different 
values of $k$ and $N$ (a priori one should expect that systematic errors 
become large for very small $N$), and for many more distributions (exponential,
uniform, ...). In all cases we found that both variants ${\hat I}^{(1)}(X,Y)$ and ${\hat I}^{(2)}(X,Y)$ 
become exact for independent variables.

\subsection{Algorithmic Information Theory}
\label{subsec:algo}
In contrast to Shannon theory where the basic objects are random variables and entropies are {\it average}
informations, algorithmic information theory deals with individual symbol strings and with the actual
information needed to specify them. To ``specify" a sequence $X$ means here to give the necessary input to a
universal computer $U$, such that $U$ prints $X$ on its output and stops. The analogon to entropy, called here
usually the {\it complexity} $K(X)$ of $X$, is the minimal length of any input which leads to the output $X$, for
fixed $U$. It depends on $U$, but it can be shown that this dependence is weak and can be neglected in the limit
when $K(X)$ is large \cite{li-vi,cover-thomas}.

Let us denote the concatenation of two strings $X$ and $Y$ as $XY$. Its complexity is $K(XY)$. It is intuitively
clear that $K(XY)$ should be larger than $K(X)$ but cannot be larger than the sum $K(X)+K(Y)$. Even if $X$ and 
$Y$ are completely unrelated, so that one cannot learn anything about $Y$ by knowing $X$, $K(XY)$ is slightly 
smaller that $K(X)+K(Y)$. The reason is simply that the information needed to reconstruct $XY$ (which is measured 
by $K(XY)$) does not include the information about where $X$ ends and $Y$ starts (which is included of course 
in $K(X)+K(Y)$). The latter information increases logarithmically with the total length $N$ of the sequence $XY$.
It is one of the sources for ubiquitous terms of order $\log(N)$ which become irrelevant in the limit $N\to\infty$, 
but make rigorous theorems in algorithmic information theory look unpleasant.

Up to such terms, $K(X)$ satisfies the following seemingly obvious but non-trivial properties \cite{cilibrasi05}:
\begin{enumerate}
\item Idempotency: $K(XX)=K(X)$
\item Monotonicity: $K(XY)\geq K(X)$
\item Symmetry: $K(XY)=K(YX)$
\item Distributivity: $K(XY)+K(Z) \leq K(XZ)+K(YZ)$
\end{enumerate}

 Finally, one
expects that $K(X|Y)$, defined as the minimal length of a program printing $X$ when $Y$ is furnished as
auxiliary input, is related to $K(XY)-K(Y)$. Indeed, one can show \cite{li-vi} (again within correction terms
which become irrelevant asymptotically) that
\be
   0 \leq K(X|Y) \simeq K(XY)-K(Y) \leq K(X).
\ee
Notice the close similarity with Shannon entropy.

The algorithmic information in $Y$ about $X$ is finally defined as
\be
   I_{\rm alg}(X,Y) = K(X) - K(X|Y) \simeq K(X)+K(Y)-K(XY).
\ee
Within the same additive correction terms, one shows that it is symmetric, $I_{\rm alg}(X,Y) = I_{\rm
alg}(Y,X)$, and can thus serve as an analogon to mutual information.

From Turing's halting theorem it follows that $K(X)$ is in general not computable. But one can easily give upper
bounds. Indeed, the length of any input which produces $X$ (e.g. by spelling it out verbatim) is an upper bound.
Improved upper bounds are provided by any file compression algorithm such as gnuzip or UNIX ``compress". Good
compression algorithms will give good approximations to $K(X)$, and algorithms whose performance does not depend
on the input file length (in particular since they do not segment the file during compression) will be crucial
for the following. As argued in \cite{cilibrasi05}, it is not necessary that the compression algorithm gives 
estimates of $K(X)$ which are close to the true values, as long as it satisfies points 1-4 above.
Such a compression algorithm is called {\it normal} in 
\cite{cilibrasi05}. While the old UNIX ``compress" algorithm is not normal (idempotency is badly violated), 
most modern compression algorithms (see \cite{lpaq} for an exhaustive overview) are close to normal. 

Before leaving this subsection, we should mention that $K(X|Y)$ can also be estimated in a completely 
different way, by aligning $X$ and $Y$ \cite{penner}. If $X$ and $Y$ are sufficiently close so that a global
alignment makes sense, one can form from such an alignment a {\it translation string} $T_{Y\to X}$ which 
is such that $Y$ and $T_{Y\to X}$ together determine $X$ uniquely. Then $K(T_{Y\to X})$ is an upper bound
to $K(X|Y)$, and can be estimated by compressing $T_{Y\to X}$. In \cite{penner} this was applied among 
others to complete mitochondrial genomes of vertebrates. The estimates obtained with state of the art global 
alignment and text compression algorithms (MAVID \cite{mavid} and lpaq1 \cite{lpaq}) were surprisingly close 
to those obtained by the compression-and-concatenation method with the gene compression algorithm XM \cite{XM}.
The latter seems at present by far the best algorithm for compressing DNA sequences. Estimates for $K(X|Y)$
obtained with other algorithms such as gencompress \cite{GenComp} were typically smaller by nearly an order of 
magnitude.

\subsection{MI-Based Distance Measures}
\label{subsec:midist}

Mutual information itself is a similarity measure in the sense that small values imply large ``distances" in a
loose sense. But it would be useful to modify it such that the resulting quantity is a metric in the strict
sense, i.e. satisfies the triangle inequality. Indeed, the first such metric is well known within Shannon 
theory \cite{cover-thomas}: The quantity
\be
   d(X,Y)=H(X|Y)+H(Y|X)=H(X,Y)-I(X,Y)                   \label{d}
\ee
satisfies the triangle inequality, in addition to being non-negative and symmetric and to satisfying $d(X,X)=0$.
The analogous statement in algorithmic information theory, with $H(X,Y)$ and $I(X,Y)$ replaced by $K(XY)$ and
$I_{\rm alg}(X,Y)$, was proven in \cite{li01,li04}.

But $d(X,Y)$ is not appropriate for our purposes. Since we want to compare the proximity between two single
objects and that between two clusters containing maybe many objects, we would like the distance measure to be
unbiased by the sizes of the clusters. As argued forcefully in \cite{li01,li04}, this is not true for $I_{\rm
alg}(X,Y)$, and for the same reasons it is not true for $I(X,Y)$ or $d(X,Y)$ either: A mutual information of
thousand bits should be considered as large, if $X$ and $Y$ themselves are just thousand bits long, but it
should be considered as very small, if $X$ and $Y$ would each be huge, say one billion bits.

As shown in \cite{li01,li04} within the algorithmic framework, one can form two different distance measures
from MI which define metrics and which are normalized. As shown in \cite{kraskov03} (see also \cite{yu07}),
the proofs of \cite{li01,li04} can be transposed almost verbatim to the Shannon case. In the following we 
quote only the latter. 

\begin{theorem} The quantity
\be
   D(X,Y) = 1 - \frac{I(X,Y)}{H(X,Y)} = \frac{d(X,Y)}{H(X,Y)}                 \label{eq:dist}
\ee
is a metric, with $D(X,X)=0$ and $D(X,Y)\leq 1$ for all pairs $(X,Y)$.
\end{theorem}

\begin{theorem}
 The quantity
\bea
   D'(X,Y) & = & 1 - \frac{I(X,Y)}{\max\{H(X),H(Y)\}}    \nonumber \\
     & = & \frac{\max\{H(X|Y),H(Y|X)\}}{\max\{H(X),H(Y)\}}                \label{eq:dist2}
\eea
is also a metric, also with $D'(X,X)=0$ and $D'(X,Y)\leq 1$ for all pairs $(X,Y)$. It is sharper than $D$, i.e.
$D'(X,Y) \leq D(X,Y)$.
\end{theorem}

Apart from scaling correctly with the total information, in contrast to $d(X,Y)$, the algorithmic analogs to
$D(X,Y)$ and $D'(X,Y)$ are also {\it universal} \cite{li04}. Essentially this means that if $X\approx Y$
according to any non-trivial distance measure, then $X\approx Y$ also according to $D$, and even more so (by
a factor up to 2) according to $D'$. In contrast to the other properties of $D$ and $D'$, this is not easy to
carry over from algorithmic to Shannon theory. The proof in Ref.~\cite{li04} depends on $X$ and $Y$ being
discrete, which is obviously not true for probability distributions. Based on the universality argument, it was
argued in \cite{li04} that $D'$ should be superior to $D$, but the numerical studies shown in that reference did
not show a clear difference between them. In addition, $D$ is singled out by a further property:

\begin{theorem} Let $X$ and $Y$ be two strings, and let $W$ be the concatenation $W=XY$. Then $W$ is a weighted 
"mid point" in the sense that 
\be
   D(X,W) + D(W,Y) = D(X,Y), \qquad D(X,W) : D(Y,W) = H(Y|X) : H(X|Y).
\ee
\end{theorem}

\begin{proof} We present the proof in its Shannon version. The algorithmic version is basically the same, 
if one neglects the standard logarithmic correction terms.

Since $H(X|W) = 0$, one has $I(X,W) = H(X)$. Similarly, $H(X,W) = H(X,Y)$. Thus 
\be
   D(X,W) + D(W,Y) = 2 - {H(X)+H(Y)\over H(X,Y)} = 1 - {I(X,Y)\over H(X,Y)} = D(X,Y),
\ee
which proofs the first part. The second part is proven similarly by straightforward calculation.
\end{proof}

For $D'$ one has only the inequalities $D'(X,Y) \leq D'(X,W)+D'(W,Y)\leq D(X,Y)$.

Theorem 3 provides a strong argument for using $D$ in MIC, instead of $D'$. This does not mean that 
$D$ is {\it always} preferable to $D'$. Indeed, we will see in the next section that MIC is not always 
the most natural clustering algorithm, but that depends very much on the application one has in mind. Anyhow, 
we found numerically that in all cases $D$ gave at least comparable results as $D'$.

A major difficulty appears in the Shannon framework, if we deal with continuous random variables. As we
mentioned above, Shannon informations are only finite for coarse-grained variables, while they diverge if the
resolution tends to zero. This means that dividing MI by the entropy as in the definitions of $D$ and $D'$
becomes problematic. One has essentially two alternative possibilities. The first is to actually introduce some
coarse-graining, although it would not have been necessary for the definition of $I(X,Y)$, and divide by the
coarse-grained entropies. This introduces an arbitrariness, since the scale $\Delta$ is completely ad hoc,
unless it can be fixed by some independent arguments. We have found no such arguments, and thus we propose the
second alternative. There we take $\Delta \to 0$. In this case $H(X) \sim m_x \log \Delta$, with $m_x$ being the
dimension of $X$. In this limit $D$ and $D'$ would tend to 1. But using similarity measures
\bea
   S(X,Y) = (1-D(X,Y))\log(1/\Delta),          \\
   S'(X,Y) = (1-D'(X,Y))\log(1/\Delta)
\eea
instead of $D$ and $D'$ gives {\it exactly} the same results in MIC, and
\be
   S(X,Y) = \frac{I(X,Y)}{m_x+m_y}, \quad S'(X,Y) = \frac{I(X,Y)}{\max\{m_x,m_y\}}.
                                         \label{S}
\ee
Thus, when dealing with continuous variables, we divide the MI either by the sum or by the maximum of the
dimensions. When starting with scalar variables and when $X$ is a cluster variable obtained by joining $m$
elementary variables, then its dimension is just $m_x=m$.

\subsection{Properties and Presentation of MIC Trees}

MIC gives rooted trees: The leaves are original sequences/variables $X,\ldots, Z$, internal nodes correspond 
to subsets of the set of all leaves, and the root represents the total set of all leaves, i.e. the joint variable 
$(X\ldots Z)$. A bad choice of the metric and/or of the clustering algorithm will in general manifest itself 
in long ``caterpillar-like" branches, while good choices will tend towards more balanced branchings. 

When presenting the tree, it is useful to put all leaves on the $x$-axis, and to use information about 
the distances/similarities to control the height. We have essentially two natural choices, illustrated 
in Figs.~\ref{fig:phylotree} and \ref{fig:ICAECG}, respectively. In Fig.~\ref{fig:ICAECG}, the height of 
a node is given by the mutual information between all leaves in the subtree below it. If the node is a leave,
its height is zero. If it is an internal node (including the root) corresponding to a subset ${\cal S}$ of 
leaves, then its height is given by the mutual information between all members of ${\cal S}$,
\be
    height({\cal S}) = I({\cal S}) \qquad {\rm method 1}.
\ee
Let us assume that ${\cal S}$ has the two daughters $X$ and $Y$, i.e. ${\cal S} = (XY)$. X and Y themselves might 
be either leaves or also internal nodes. Then the grouping property implies that 
$ height({\cal S}) - height(X) = I({\cal S})-I(X)$ is the MI between $X$ and all the other sequences/variables 
in ${\cal S}$ which are not contained in $X$. This is non-negative by the positivity of MI, and the same is true 
when $X$ is exchanged with $Y$. Thus the trees 
drawn in this way are always {\it well formed} in the sense that a mother node is located above its daughters.
Any violation of this rule must be due to imprecise estimation of MI's.

A drawback of this method of representing the tree is that very close pairs have long branches, while relatively
distant pairs are joined by short branches. This is the opposite of what is usually done in phylogenetic trees, 
where the branch lengths are representative of the distance between a mother and its daughter. This can be achieved 
by using the alternative representation employed in Fig.~\ref{fig:phylotree}. There, the height of a mother node $W$
which has two daughters $X$ and $Y$ is given by 
\be
    height(W) = D(X,Y)          \qquad {\rm method 2}.
\ee
Although it gives rise to more intuitive trees, it also has one disadvantage: It is no longer guaranteed that the 
tree is well formed, but it may happen that $height(W) < height(X)$. To see this, consider a tree formed by three
variables $X,Y$, and $Z$ which are pairwise independent but globally dependent: $I(X,Y)=I(X,Z)=I(Y,Z)=0$ but 
$I(X,Y,Z)>0$\footnote{An example is provided by three binary random variables with 
$p_{000}=p_{011}=p_{101}=p_{110}=1/2+\epsilon$ and $p_{001}=p_{010}=p_{100}=p_{111}=1/2-\epsilon$.}. 
In this case, all pairwise distances are maximal, thus also the 
first pair to be joined has distance 1. But the height of the root is less than 1. In our numerical applications
we indeed found occasionally such ``glitches", but they were rare and were usually related to imprecise 
estimates of MI or to improper choices of the metric.

\section{Mitochondrial DNA and a Phylogenetic Tree for Mammals}
\label{sec:phylo}

As a first application, we study the mitochondrial DNA of a group of 34 mammals (see Fig.~\ref{fig:phylotree}). 
Exactly the same species \cite{Genebank} had previously been analyzed in \cite{li01,reyes00,kraskov03}. 
This group includes non-eutherians\footnote{opossum (\textit{Didelphis virginiana}), 
wallaroo (\textit{Macropus robustus}), and platypus (\textit{Ornithorhyncus anatinus})}, 
rodents and lagomorphs\footnote{rabbit (\textit{Oryctolagus cuniculus}), guinea pig
(\textit{Cavia porcellus}), fat dormouse (\textit{Myoxus glis}), rat (\textit{Rattus norvegicus}), squirrel
(Sciurus vulgaris), and mouse (\textit{Mus musculus})}, ferungulates\footnote{horse (\textit{Equus caballus}),
donkey (\textit{Equus asinus}), Indian rhinoceros (\textit{Rhinoceros unicornis}), white rhinoceros
(\textit{Ceratotherium simum}), harbor seal (\textit{Phoca vitulina}), grey seal (\textit{Halichoerus grypus}),
cat (\textit{Felis catus}), dog (\textit{Canis lupus familiaris}), fin whale (\textit{Balaenoptera physalus}), blue
whale (\textit{Balenoptera musculus}), cow (\textit{Bos taurus}), sheep (\textit{Ovis aries}), pig (\textit{Sus
scrofa}) and hippopotamus (\textit{Hippopotamus amphibius})}, primates\footnote{human (\textit{Homo sapiens}), 
common chimpanzee
(\textit{Pan troglodytes}), pigmy chimpanzee (\textit{Pan paniscus}), gorilla (\textit{Gorilla gorilla}),
orangutan (\textit{Pongo pygmaeus}), gibbon (\textit{Hylobates lar}), and baboon (\textit{Papio hamadryas})},
members of the African clade\footnote{African elephant (\textit{Loxodonta africana}), aardvark (\textit{Orycteropus afer})},
and others\footnote{Jamaican fruit bat (\textit{Artibeus jamaicensis}), armadillo (\textit{Dasypus novemcintus})}. It
had been chosen in \cite{li01} because of doubts about the relative closeness among these three groups
\cite{cao98,reyes00}.

Obviously, we are here dealing with the algorithmic version of information theory, and informations are
estimated by lossless data compression. For constructing the proximity matrix between individual taxa, we
proceed essentially a in Ref.~\cite{li01}. But in addition to using the special compression program GenCompress
\cite{GenComp}, we also tested several general purpose compression programs such as BWTzip, the
UNIX tool bzip2, and lpaq1 \cite{lpaq}, and durilca \cite{lpaq}. Finally, we also tested XM \cite{XM} (the 
abbreviation stands for ``expert model"), which is according to its authors the most efficient compressor 
for biological (DNA, proteins) sequences. Indeed, we found that XM was even better than expected. While 
the advantage over GenCompress and other compressors was a few per cent when applied to 
single sequences \cite{XMresults}, the estimates for MI between not too closely related species were higher 
by up to an order of magnitude. This is possible mainly because the MI's estimated by means of GenCompress
and similar methods are reliably positive (unless the species are from different phyla) but extremely small. 
Thus even a small improvement on the compression rate can make a big change in MI.

In Ref.~\cite{li01}, the proximity matrix derived from MI estimates was then used as the input to a standard HC algorithm
(neighbor-joining and hypercleaning) to produce an evolutionary tree. It is here where our treatment deviates
crucially. We used the MIC algorithm described in Sec.~\ref{sec:intro}, with distance $D(X,Y)$. The joining of two clusters
(the third step in the MIC algorithm) is obtained by simply concatenating the DNA sequences. There is of course
an arbitrariness in the order of concatenation sequences: $XY$ and $YX$ give in general compressed sequences of
different lengths. But we found this to have negligible effect on the evolutionary tree. The resulting
evolutionary tree obtained with Gencompress is shown in Fig.~\ref{fig:phylotree}, while the tree obtained 
with XM is shown in Fig.~\ref{fig:XMtree}.

\begin{figure}
  \begin{center}
   \psfig{file=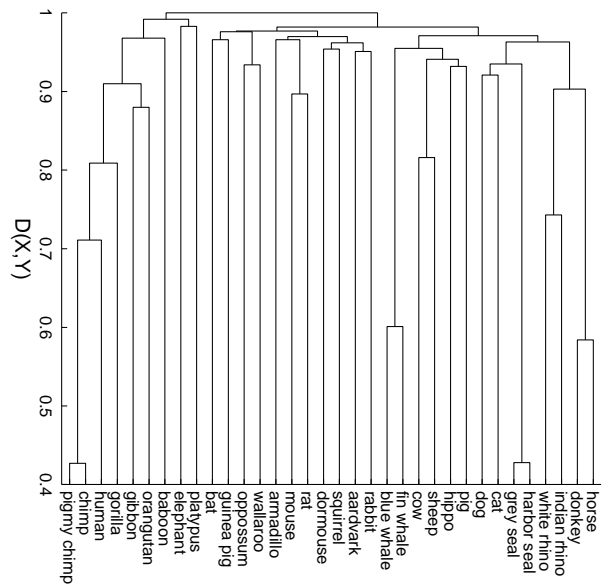,height=80mm,angle=270}
    \caption{Phylogenetic tree for 34 mammals (31 eutherians plus 3 non-placenta mammals), with mutual 
       informations estimated by means of GenCompress.
       In contrast to Fig.~\ref{fig:ClustECG}, the heights of nodes are here and in the following tree 
       equal to the distances between the joining daughter clusters.}
    \label{fig:phylotree}
\end{center}
\vspace{-8mm}
\end{figure}

\begin{figure}
  \begin{center}
   \psfig{file=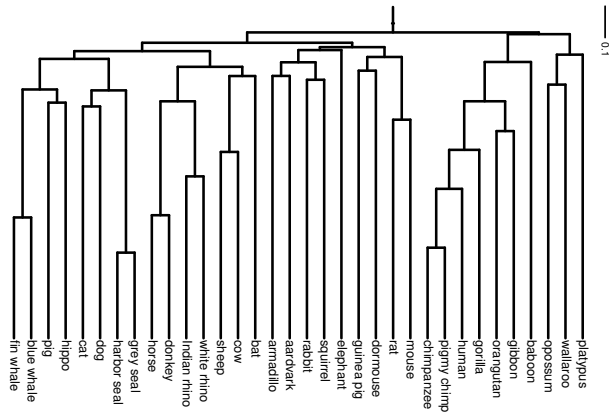,height=80mm,angle=270}
    \caption{Same as in Fig.~\ref{fig:phylotree}, but with mutual informations estimated by means of XM.
    Notice that the x-axis covers here, in contrast to Fig.~\ref{fig:phylotree}, the entire interval 
    from 0 to 1.}
    \label{fig:XMtree}
\end{center}
\vspace{-8mm}
\end{figure}

Both trees are quite similar, and they are also similar to the most widely accepted phylogeny found e.g. 
in \cite{Genebank}. All primates are e.g. correctly clustered and the ferungulates are joined together. 
There are however a number connections (in both trees) which obviously do not reflect the true evolutionary
tree. , 
As shown in Fig.~\ref{fig:phylotree} the overall structure of this tree closely resembles the one shown in
Ref.~\cite{reyes00}. All primates are correctly clustered and also the relative order of the ferungulates is in
accordance with Ref.~\cite{reyes00}. On the other hand, there are a number of connections which obviously do not
reflect the true evolutionary tree, see for example the guinea pig with bat and elephant with platypus in 
Fig.~\ref{fig:phylotree}, and the mixing of rodents with African clade and armadillo in Fig.~\ref{fig:XMtree}. But 
all these wrong associations are between species which have a very large distance from each other and from
any other species within this sample. All in all, the results shown in Figs.~\ref{fig:phylotree} and 
\ref{fig:phylotree} are in surprisingly good agreement (showing that neither compression scheme has obvious 
faults) and capture surprisingly well the known relationships between mammals. Dividing MI by the total information
is essential for this success. If we had used the non-normalized $I_{\rm alg}(X,Y)$ itself, results obtained 
in \cite{li01} would not change much, since all 34 DNA sequences have roughly the same length. 
But our MIC algorithm would be completely screwed up:
After the first cluster formation, we have DNA sequences of very different lengths, and longer sequences tend
also to have larger MI, even if they are not closely related.

One recurrent theme in the discussion of mammalian phylogenetic trees is the placement of the rodents 
\cite{reyes00,reyes04}. 
Are they closer to ferungulates, or closer to primates? Our results are inconclusive. On the one hand,
the average distances between all 14 rodents and all 104 ferungulates in the Genebank data (Feb. 2008), 
estimated with XM, is 0.860 -- which is clearly smaller than the average distance 0.881 to all 28 primates. 
On the other hand, the distances within the group of rodents are very large, suggesting that this group is 
either not monophylic, or that its mtDNA has evolved much faster than, say, that of ferungulates. Either 
possibility makes a statement about the classification of rodents with respect to ferungulates and 
primates based on these data very uncertain.

A heuristic reasoning for the use of MIC for the reconstruction of an evolutionary tree might be given as
follows: Suppose that a proximity matrix has been calculated for a set of DNA sequences and the smallest
distance is found for the pair $(X,Y)$. Ideally, one would remove the sequences $X$ and $Y$, replace them by the
sequence of the common ancestor (say $Z$) of the two species, update the proximity matrix to find the smallest
entry in the reduced set of species, and so on. But the DNA sequence of the common ancestor is not available.
One solution might be that one tries to reconstruct it by making some compromise between the sequences $X$ and
$Y$. Instead, we essentially propose to concatenate the sequences $X$ and $Y$. This will of course not lead to a
plausible sequence of the common ancestor, but it will {\it optimally represent the information} about the
common ancestor. During the evolution since the time of the ancestor $Z$, some parts of its genome might have
changed both in $X$ and in $Y$. These parts are of little use in constructing any phylogenetic tree. Other parts
might not have changed in either. They are recognized anyhow by any sensible algorithm. Finally, some parts of
its genome will have mutated significantly in $X$ but not in $Y$, and vice versa. This information is essential
to find the correct way through higher hierarchy levels of the evolutionary tree, and it is preserved in
concatenating.

In any case, this discussion shows that our clustering algorithm produces trees which is closer in spirit to 
phenetic trees than to phylogenetic trees proper. As we said, in the latter the internal nodes 
should represent actual extinct species, namely the last common ancestors. In our method, in contrast,
an internal node does not represent a particular species but a higher order clade which is defined solely 
on the basis of information about presently living species. In the phylogenetic context, purely phenetic
trees are at present much less popular than trees using evolutionary information. This is not so in the 
following application, where no evolutionary aspect exists and the above discussion is irrelevant.

\section{Clustering of Minimally Dependent Components in an Electrocardiogram}
\label{sec:ecg}
As our second application we choose a case where Shannon theory is the proper setting. We show in Fig.~\ref{fig:ECGdata} an ECG
recorded from the abdomen and thorax of a pregnant woman \ref{fig:ECGdata} (8 channels, sampling rate 500 Hz,
5$\,$s total). It is already seen from this graph that there are at least two important components in this ECG:
the heartbeat of the mother, with a frequency of $\approx 3$ beat/s, and the heartbeat of the fetus with roughly
twice this frequency. Both are not synchronized. In addition there is noise from various sources (muscle
activity, measurement noise, etc.). While it is easy to detect anomalies in the mother's ECG from such a
recording, it would be difficult to detect them in the fetal ECG.

As a first approximation we can assume that the total ECG is a linear superposition of several independent
sources (mother, child, noise$_1$, noise$_2$,...). A standard method to disentangle such superpositions is {\it
independent component analysis} (ICA) \cite{ICA}. In the simplest case one has $n$ independent sources
$s_i(t),\; i=1\ldots n$ and $n$ measured channels $x_i(t)$ obtained by instantaneous superpositions with a time
independent non-singular matrix ${\bf A}$,
\be
   x_i(t) = \sum_{j=1}^n A_{ij} s_j(t)\;.
\ee
In this case the sources can be reconstructed by applying the inverse transformation ${\bf W} = {\bf A}^{-1}$
which is obtained by minimizing the (estimated) mutual informations between the transformed components $y_i(t) =
\sum_{j=1}^n W_{ij} x_j(t)$. If some of the sources are Gaussian, this leads to ambiguities \cite{ICA}, but it
gives a unique solution if the sources have more structure.

In reality things are not so simple. For instance, the sources might not be independent, the number of sources
(including noise sources!) might be different from the number of channels, and the mixing might involve delays.
For the present case this implies that the heartbeat of the mother is seen in several reconstructed components
$y_i$, and that the supposedly ``independent" components are not independent at all. In particular, all
components $y_i$ which have large contributions from the mother form a cluster with large intra-cluster MIs and
small inter-cluster MIs. The same is true for the fetal ECG, albeit less pronounced.
It is thus our aim to \\
1) optimally decompose the signals into least dependent components;\\
2) cluster these components hierarchically such that the most dependent ones are
grouped together;\\
3) decide on an optimal level of the hierarchy, such that the clusters make most sense
physiologically;\\
4) project onto these clusters and apply the inverse transformations to obtain cleaned signals for the sources
of interest.

\begin{figure}
  \begin{center}
    \psfig{file=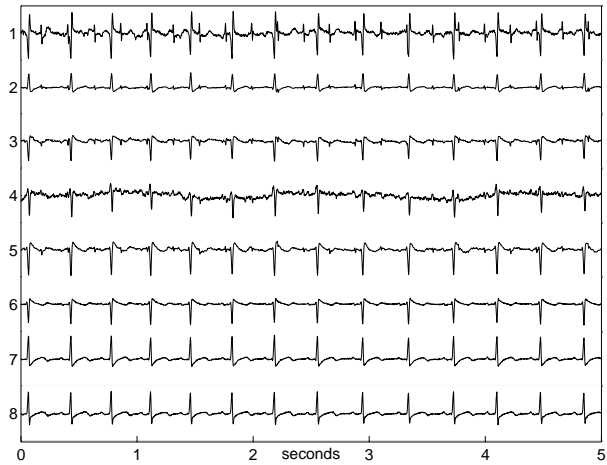,width=80mm}
    \caption{ECG of a pregnant woman.}
    \label{fig:ECGdata}
\end{center}
\vspace{-8mm}
\end{figure}
\begin{figure}
  \begin{center}
    \psfig{file=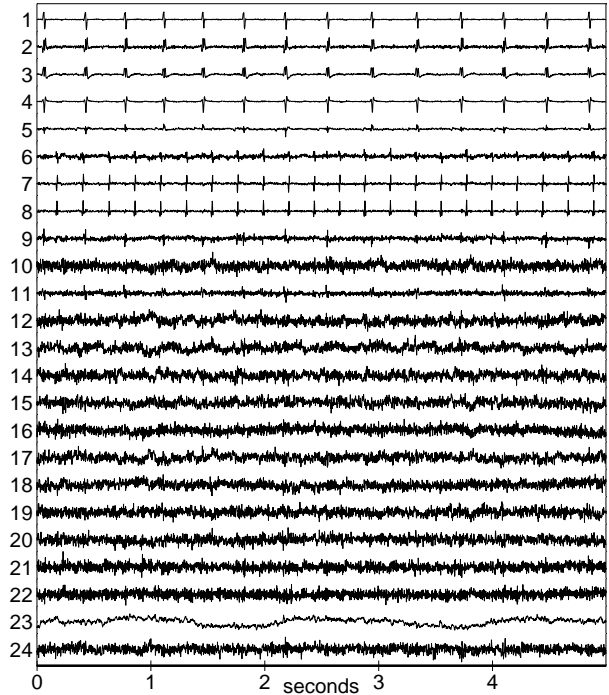,width=80mm}
    \caption{Least dependent components of the ECG shown in Fig.~\ref{fig:ECGdata}, after increasing
     the number of channels by delay embedding.}
    \label{fig:ICAECG}
\end{center}
\vspace{-8mm}
\end{figure}

Technically we proceeded as follows \cite{stoegbauer04}:

Since we expect different delays in the different channels, we first used Takens delay embedding \cite{takens80}
with time delay 0.002$\,$s and embedding dimension 3, resulting in $24$ channels. We then formed 24 linear
combinations $y_i(t)$ and determined the de-mixing coefficients $W_{ij}$ by minimizing the overall mutual
information between them, using the MI estimator proposed in \cite{kraskov04}. There, two classes of estimators were
introduced, one with square and the other with rectangular neighborhoods. Within each class, one can use the
number of neighbors, called $k$ in the following, on which the estimate is based. Small values of $k$ lead to a
small bias but to large statistical errors, while the opposite is true for large $k$. But even for very large
$k$ the bias is zero when the true MI is zero, and it is systematically such that absolute values of the MI are
underestimated. Therefore this bias does not affect the determination of the optimal de-mixing matrix. But it
depends on the dimension of the random variables, therefore large values of $k$ are not suitable for the
clustering. We thus proceeded as follows: We first used $k=100$ and square neighborhoods to obtain the least
dependent components $y_i(t)$, and then used $k=3$ with rectangular neighborhoods for the clustering. The
resulting least dependent components are shown in Fig.~\ref{fig:ICAECG}. They are sorted such that the first
components (1 - 5) are dominated by the maternal ECG, while the next three contain large contributions from the
fetus. The rest contains mostly noise, although some seem to be still mixed.

These results obtained by visual inspection are fully supported by the cluster analysis. The dendrogram is shown
in Fig.~\ref{fig:ClustECG}. In constructing it we used $S(X,Y)$ (Eq.(\ref{S})) as similarity measure to find the
correct topology. Again we would have obtained much worse results if we had not normalized it by dividing MI by
$m_X+m_Y$. In plotting the actual dendrogram, however, we used the MI of the cluster to determine the height at
which the two daughters join. The MI of the first five channels, e.g., is $\approx 1.43$, while that of channels
6 to 8 is $\approx 0.34$. For any two clusters (tuples) $X=X_1\ldots X_n$ and $Y=Y_1\ldots Y_m$ one has $I(X,Y)
\geq I(X)+I(Y)$. This guarantees, if the MI is estimated correctly, that the tree is drawn properly. The two
slight glitches (when clusters (1--14) and (15--18) join, and when (21--22) is joined with 23) result from small
errors in estimating MI. They do in no way effect our conclusions.
\begin{figure}
  \begin{center}
    \psfig{file=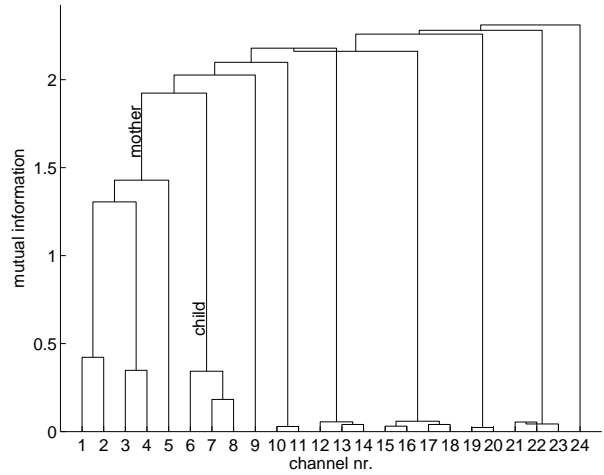,width=80mm,angle=0}
    \caption{Dendrogram for least dependent components. The height where the two branches of
    a cluster join corresponds to the MI of the cluster.}
    \label{fig:ClustECG}
\end{center}
\vspace{-8mm}
\end{figure}

In Fig.~\ref{fig:ClustECG} one can clearly see two big clusters corresponding to the mother and to the child. There
are also some small clusters which should be considered as noise. For reconstructing the mother and child
contributions to Fig.~\ref{fig:ECGdata}, we have to decide on one specific clustering from the entire hierarchy. We
decided to make the cut such that mother and child are separated. The resulting clusters are indicated in
Fig.~\ref{fig:ClustECG} and were already anticipated in sorting the channels. Reconstructing the original ECG from
the child components only, we obtain Fig.~\ref{fig:reconstruct}.

\begin{figure}
  \begin{center}
    \psfig{file=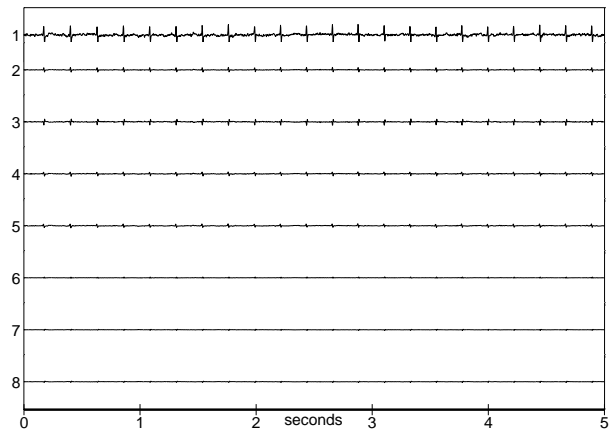,width=80mm}
    \caption{Original ECG where all contributions except those of the child cluster have
     been removed.}
    \label{fig:reconstruct}
\vspace{-8mm}
\end{center}
\end{figure}

\section{Conclusions}
\label{sec:concl}

We have shown that MI can not only be used as a proximity measure in clustering, but that it also suggests a
conceptually very simple and natural hierarchical clustering algorithm. We do not claim that this algorithm,
called {\it mutual information clustering} (MIC), is always superior to other algorithms. Indeed, MI is in
general not easy to estimate. Obviously, when only crude estimates are possible, also MIC will not give optimal
results. But as MI estimates are becoming better, also the results of MIC should improve. The present paper was
partly triggered by the development of a new class of MI estimators for continuous random variables which have
very small bias and also rather small variances \cite{kraskov04}.

We have illustrated our method with two applications, one from genetics and one from cardiology. For neither
application MIC might give the very best clustering, but it seems promising and indicative of the inherit 
simplicity of our method that one common method gives decent results in two very different applications.

If better data become available, e.g. in the form of longer time sequences in the application to ECG or of 
more complete genomes (so that mutual information can be estimated more reliably), then the results of MIC 
should improve. It is not obvious what to expect when one wants to include more data in order to estimate
larger trees. On the one hand, more species within one taxonomic clade would describe this clade more 
precisely, so results should improve. On the other hand, as clusters become bigger and bigger, also the
disparities of the sequences lengths which describe these clusters increase. It is not clear whether 
in this case a strict normalization of the distances as in E's.~(\ref{eq:dist},\ref{eq:dist2}) is still appropriate, 
and whether the available compression algorithms will still be able to catch the very long resulting 
correlations within the concatenated sequences. Experiments with phylogenetic trees of animals with up to 
360 species (unpublished results) had mixed success. 

As we said in the introduction, after a construction of a first tree one can try to improve on it. One 
possibility is to change the topology of the tree, using e.g. quartet moves and accepting them based on
some heuristic cost function. One such cost function could be the sum of all distances between linked 
nodes in the tree. Alternatively, one could try to keep the topology fixed and change the sequences 
representing the internal nodes, i.e. deviate from simple concatenation. We have not tried either.

There are two versions of information theory, algorithmic and probabilistic, and therefore there are also two
variants of MI and of MIC. We discussed in detail one application of each, and showed that indeed common
concepts were involved in both. In particular it was crucial to normalize MI properly, so that it is essentially
the {\it relative} MI which is used as proximity measure. For conventional clustering algorithms using
algorithmic MI as proximity measure this had already been stressed in \cite{li01,li04}, but it is even more
important in MIC, both in the algorithmic and in the probabilistic versions.

In the probabilistic version, one studies the clustering of probability distributions. But usually distributions
are not provided as such, but are given implicitly by finite random samples drawn (more or less) independently
from them. On the other hand, the full power of algorithmic information theory is only reached for infinitely
long sequences, and in this limit any individual sequence defines a sequence of probability measures on finite
subsequences. Thus the strict distinction between the two theories is somewhat blurred in practice.
Nevertheless, one should not confuse the similarity between two sequences (two English books, say) and that
between their subsequence statistics. Two sequences are maximally different if they are completely random, but
their statistics for short subsequences is then identical (all subsequences appear in both with equal
probabilities). Thus one should always be aware of what similarities or independencies one is looking for. The
fact that MI can be used in similar ways for all these problems is not trivial.

\begin{acknowledgement}
This work was done originally in collaboration with H. St\"ogbauer and R.G. Andrzejak. We are very much 
indebted to them for their contributions. We also would like to thank Arndt von Haeseler, Walter Nadler, Volker Roth, 
Orion Penner and Maya Pczuski for many useful and critical discussions.
\end{acknowledgement}
\bibliographystyle{spmpsci}
\bibliography{chapter}
\end{document}